\documentclass{llncs}
\usepackage{amsmath}
\usepackage{amssymb}
\usepackage{amsfonts}
\usepackage{color}
\usepackage{xspace}

\newcommand{\ignore}[1]{}
\renewcommand{\Pr}{{\bf Pr}}
\newcommand{\E}{{\bf E}}

\newcommand{\A}{{\cal A}}

\newcommand{\OO}{{\cal O}}

\begin{document}
\title{Adaptive Group Testing Algorithms to \\Estimate the Number of Defectives}
%\author{Nader H. Bshouty}
%\institute{Technion, Haifa, Israel\\
%bshouty@ca.technion.ac.il}
\author{{\bf Nader H. Bshouty}\\ Dept. of Computer Science\\ Technion,  Haifa, 32000\\
{\bf Vivian E. Bshouty-Hurani} \\  The Arab Orthodox Collage \\  Haifa\\
{\bf George Haddad}  \\  The Arab Orthodox Collage \\ Grade 10\\ Haifa\\
{\bf Thomas Hashem} \\ Sister of St. Joseph School \\
Grade 11\\ Nazareth\\
{\bf Fadi Khoury}\\ Sister of Nazareth High School\\ Grade 10\\ P.O.B. 9422, Haifa, 35661\\
{\bf Omar Sharafy}\\ The Arab Orthodox Collage \\ Grade 10\\Haifa
}
\institute{}
\maketitle

\begin{abstract}
We study the problem of estimating the number of defective items in adaptive Group testing by using a minimum number of queries. We improve the existing algorithm and prove a lower bound that shows that, for constant estimation, the number of tests in our algorithm is optimal.
\end{abstract}
\section{Introduction}
Let $X$ be a set of items, among which some are defective, denoted as $I \subseteq X$. In group testing, we test (or query) a subset $Q \subset X$ of items. The response to the query is '1' if $Q$ contains at least one defective item, i.e., $Q \cap I \neq \emptyset$, and '0' otherwise.

Group testing was initially introduced as a method for economical mass blood testing \cite{D43}. Since then, its applicability has extended to various problems, such as DNA library screening \cite{ND00}, quality control in product testing \cite{SG59}, file searching in storage systems \cite{KS64}, sequential screening of experimental variables \cite{L62}, developing efficient contention resolution algorithms for multiple-access communication \cite{KS64,W85}, data compression \cite{HL02}, and computation in the data stream model \cite{CM05}. For a brief history and additional applications, see \cite{Ci13,DH00,DH06,H72,MP04,ND00} and the references therein.

Estimating the number of defective items, $|I|$, within a multiplicative factor of $1 \pm \epsilon$ has been studied in various works \cite{ChengX14,DamaschkeM10,DamaschkeM10b,FalahatgarJOPS16,RonT14}. This estimation is crucial in biological and medical applications \cite{ChenS90,Swallow85}. For instance, it is used to determine the proportion of leafhoppers capable of transmitting the aster-yellows virus in their natural population \cite{Thompson62}, to estimate the infection rate of the yellow-fever virus in mosquito populations \cite{WalterHB80}, and to assess the prevalence of rare diseases using grouped samples, which helps in maintaining individual anonymity \cite{GatwirthH89}.

In the {\it adaptive algorithm}, the tests can depend on the answers to the previous ones. In the {\it non-adaptive algorithm}, they are independent of the previous one; therefore, all the tests can be done in one parallel step.

In this paper, we explore the problem of estimating the number of defective items, denoted as $|I|$, up to a multiplicative factor of $1 \pm \epsilon$ using adaptive group testing algorithms. We first present new lower bounds and then introduce algorithms that enhance the results found in existing literature. Our lower bounds demonstrate the optimality of our algorithms.

\subsection{Previous and New Results}\label{sectionTable}
Let $X$ be a set of $n$ items with a subset of defective items $I$.
Estimating the number of defective items, $|I|=d$, up to a multiplicative factor of $1\pm \epsilon$ has been studied in~\cite{ChengX14,DamaschkeM10,DamaschkeM10b,FalahatgarJOPS16,RonT14}. The most efficient algorithm to date is that of Falahatgar et al. \cite{FalahatgarJOPS16}. They presented a randomized algorithm that asks $2\log\log d+O((1/\epsilon^2)\log(1/\delta))$ {\it expected number of queries} and, with probability at least $1-\delta$, returns an estimation of $d$ within a multiplicative factor of $1 \pm \epsilon$. They also established a lower bound of $(1-\delta)\log\log d$. We show that, with certain modifications to their algorithm, one can get the same result with $(1-\delta)\log\log d+O((1/\epsilon^2)\log(1/\delta))$ expected number of queries. We further establish a lower bound of $(1-\delta)\log \log d+(1/\epsilon)\log(1/\delta)$ for the number of queries. This indicates that our algorithm is optimal for constant $\epsilon$.

While improvements in constant factors in group testing algorithms may seem minor at first glance, they are, in fact, of paramount importance in this field. This is because, in many applications, queries are incredibly costly and time-consuming.

The randomized algorithms mentioned above are not Monte Carlo algorithms. They are characterized by their expected query complexity. We further investigate randomized Monte Carlo, deterministic, and randomized Las Vegas randomized algorithms for this problem. For the randomized Monte Carlo algorithms, we establish a lower bound of $\log\log d + \frac{1}{\epsilon}\log\left(\frac{1}{\delta}\right)$. Subsequently, we present an algorithm that requires $\log\log n + O\left(\frac{1}{\epsilon^2}\log\left(\frac{1}{\delta}\right)\right)$ queries.

For both deterministic and randomized Las Vegas algorithms, we prove the lower bound of $d\log \left(\frac{(1-\epsilon)n}{d}\right)$. Subsequently, we introduce a deterministic algorithm whose number of queries matches this lower bound.

All the algorithms mentioned above run in linear time with respect to $n$. The table below summarizes our results:

\begin{center}
\begin{tabular}{l||c|c|}
Adaptive Algorithm & Upper Bound & Lower Bound  \\
 \hline\hline
Deterministic  & $d\log\frac{(1-\epsilon)n}{d}$ & $d\log\frac{(1-\epsilon)n}{d}$\\
\hline
Randomized Las Vegas  & $d\log\frac{(1-\epsilon)n}{d}$ & $d\log\frac{(1-\epsilon)n}{d}$\\
\hline
Randomized Monte Carlo  & $\log\log n+O\left(\frac{1}{\epsilon^2}\log\frac{1}{\delta}\right)$ & $\log\log d+\Omega\left(\frac{1}{\epsilon}\log\frac{1}{\delta}\right)$\\
\hline
Randomized Monte Carlo  &$(1-\delta)\log\log d+$  & $(1-\delta)\log\log d+$\\
With Expected $\#$Queries   & $O\left(\frac{1}{\epsilon^2}\log\frac{1}{\delta}\right)$ & $\Omega\left(\frac{1}{\epsilon}\log\frac{1}{\delta}\right)$\\
\hline
\hline
\end{tabular}
\end{center}

All the algorithms in this paper are {\it adaptive}. That is, the tests can depend on the answers to the previous ones. For studies on non-adaptive algorithms, refer to~\cite{DamaschkeM10,DamaschkeM10b}. For an algorithm that accurately determines the number of defective items, see~\cite{Cheng11}. The most efficient adaptive algorithm for identifying the defective items requires $d\log(n/d)+O(d)$ queries~\cite{ChengDX14,ChengDZ15,SchlaghoffT05}. This query complexity matches the information-theoretic lower bound applicable to any deterministic or randomized algorithm.

\section{Definitions and Preliminary Results}
In this section, we give some notations, definitions, the type of algorithms that are used in the literature, and some preliminary results.
\subsection{Notations and Definitions}
Let $X=[n]:=\{1,2,3,,\ldots,n\}$ be a set of {\it items} with some {\it defective} items $I\subseteq [n]$. In Group testing, we {\it query} a subset $Q\subseteq X$ of items, and the answer to the query is $Q(I):=1$ if $Q$ contains at least one defective item, i.e., $Q\cap I\not=\O$, and $Q(I):=0$, otherwise.

Let $I\subseteq [n]$ be the set of defective items. Let $\OO_I$ be an {\it oracle} that for a query $Q\subseteq [n]$ returns $Q(I)$. Let $A$ be an algorithm with access to the oracle $\OO_I$. The output of the algorithm $A$ for an oracle $\OO_I$ is denoted by $A(\OO_I)$. When the algorithm is randomized, we add the random seed $r$ as an input to $A$, and then the output of the algorithm is a random variable $A(\OO_I,r)$ in $[n]$. Let $A$ be a randomized algorithm and $r_0$ be a seed. We denote by $A(r_0)$ the deterministic algorithm that is equivalent to the algorithm $A$ with the seed $r_0$. We denote by $Q(A,\OO_I)$ (resp., $Q(A(r),\OO_I)$) the set of queries that $A$ asks with oracle $\OO_I$ (resp., and a seed $r$). The algorithms we consider in this paper output $A(\OO_I,r)\in [|I|(1-\epsilon),|I|(1+\epsilon)]$ where $[a,b]=\{\lceil a\rceil ,\lceil a\rceil+1,\cdots,\lfloor b\rfloor\}$. Such algorithms are called algorithms {\it that estimate the number of defective items $|I|$ up to a multiplicative factor of $1\pm \epsilon$}.

\subsection{Type of Algorithms}
In this paper, we consider four types of algorithms that estimate the number of defective items $|I|$ up to a multiplicative factor of $1\pm \epsilon$.
\begin{enumerate}
\item The {\it deterministic} algorithm $A$ with an oracle $\OO_I$, $I\subseteq X$. The query complexity of a deterministic algorithm $A$ is the {\it worst case complexity}, i.e, $\max_{|I|\le d}|Q(A,\OO_I)|$.
\item The randomized Las Vegas algorithm. We say that a randomized algorithm $A$ is a {\it randomized Las Vegas algorithm that has expected query complexity} $g(n,d)$ if for any $d\in [n]$ and any $I\subseteq X$, $|I|\le d$, algorithm $A$ with an oracle $\OO_I$ asks at most $g(n,d)$ expected number of queries and with probability $1$ outputs an integer in $[|I|(1-\epsilon),|I|(1+\epsilon)]$.
\item The randomized Monte Carlo algorithm.
We say that a randomized algorithm $A$ is a {\it randomized Monte Carlo algorithm that has query complexity} $g(n,d,\delta)$ if for any $d\in [n]$ and any $I\subseteq X$, $|I|\le d$, the algorithm $A$ with an oracle $\OO_I$ asks at most $g(n,d,\delta)$ queries and with probability at least $1-\delta$ outputs an integer in~$[|I|(1-\epsilon),|I|(1+\epsilon)]$.
\item The randomized Monte Carlo algorithm with expected complexity. We say that a randomized algorithm $A$ is {\it a randomized Monte Carlo algorithm with expected complexity that has expected query complexity} $g(d,\delta)$ if, for any $d\in [n]$ and any $I\subseteq X$, $|I|\le d$, the algorithm $A$ asks $g(n,d,\delta)$ expected number of queries and with probability at least $1-\delta$ outputs an integer in $[|I|(1-\epsilon),|I|(1+\epsilon)]$.
\end{enumerate}

\subsection{Preliminary Results}
We now provide a few results that will be used throughout the paper

Let $s\in \cup_{i=0}^\infty\{0,1\}^i$ be a {\it string} over $\{0,1\}$ (including the empty string $\lambda\in \{0,1\}^0$). We denote by $|s|$ the {\it length} of $s$, i.e., the integer $m$ such that $s\in \{0,1\}^m$. Let $s_1,s_2\in \cup_{i=0}^\infty\{0,1\}^i$ be two strings over $\{0,1\}$ of lengths $m_1$ and $m_2$, respectively. We say that $s_1$ is a (proper) {\it prefix} of $s_2$ if $m_1< m_2$ and $s_{1,i}=s_{2,i}$ for all $i=1,\ldots,m_1$. We denote by $s_1\cdot s_2$ the {\it concatenation} of the two strings $s_1$ and $s_2$.

We now prove
\begin{lemma}\label{string1}\label{string2} Let $S=\{s_1,\ldots ,s_N\}$ be a set of $N$ distinct strings over $\{0,1\}$ such that no string is a prefix of another. Then, over the uniform distribution,
$$\max_{s\in S}|s|\ge E(S):=\E_{s\in S}[|s|]\ge \log N.$$
\end{lemma}
\begin{proof}The proof is by induction on $N$. For $N=1$ the set $S$ with the smallest $E(S)$ is when $S=\{\lambda\}$ and $E(S)=0=\log N$. For $N=2$ the smallest $E(S)$ is when $S=\{0,1\}$ and $E(S)=1=\log N$.
Therefore, the statement of the lemma is true for $N=1,2$.

Consider a set $S$ of size $N>2$. Obviously, $\lambda\not\in S$. Let $w\in \cup_{i=0}^\infty\{0,1\}^i$ be the longest string that is a prefix of all the strings in $S$. For $\sigma\in\{0,1\}$, let $S_\sigma=\{u\ |\ w\cdot \sigma\cdot u\in S\}$. Let $N_\sigma=|S_\sigma|$ for $\sigma\in\{0,1\}$. Obviously, $N_0+N_1=N$ and for each $\sigma\in \{0,1\}$, no string in $S_\sigma$ is a prefix of another (in $S_\sigma$). Also, $N_0,N_1>0$, because otherwise, either $w$ is not the longest common prefix of all the strings in $S$ or $w\in S$ is a prefix of another string in $S$. Let $p=N_0/N$. By the definition of $E(S)$ and the induction hypothesis
\begin{eqnarray*}
E(S)&=&|w|+1+\frac{N_0E(S_0)+N_1E(S_1)}{N}\\
&\ge& 1+\frac{N_0\log(N_0)+N_1\log(N_1)}{N}\\
&=& 1+\log(N)+p\log p+(1-p)\log(1-p)\ge \log(N).\qed
\end{eqnarray*}
\end{proof}

\begin{lemma}\label{Trivial} Let $A$ be a deterministic adaptive algorithm that asks queries and outputs an element in $[n]$. Let $I,J\subseteq X$. If
$A(\OO_I)\not= A(\OO_J)$ then there is $Q_0\in Q(A,\OO_I)\cap Q(A,\OO_J)$ such that $Q_0(I)\not=Q_0(J)$.
\end{lemma}
\begin{proof} Consider the sequence of queries $Q_{1,1}, Q_{1,2}, \cdots$ that $A$ asks with the oracle $\OO_I$
and the sequence of queries $Q_{2,1}, Q_{2,2}, \cdots$ that $A$ asks with the oracle $\OO_J$. Since $A$ is deterministic,
$A$ asks the same queries as long as it gets the same answers to the queries. That is, if $Q_{1,i}(I)=Q_{2,i}(J)$ for all $i\le \ell$ then
$Q_{1,\ell+1}=Q_{2,\ell+1}$. Since $A(\OO_I)\not= A(\OO_J)$, there must be a query $Q_0:=Q_{1,t}=Q_{2,t}$ for which $Q_0(I)\not=Q_0(J)$.\qed
\end{proof}

\begin{lemma}\label{mainlemma} Let $A$ be a deterministic adaptive algorithm that asks queries. Let $C\subseteq 2^{[n]}:=\{I|I\subseteq [n]\}$.
If for every two distinct $I_1$ and $I_2$ in $C$ there is a query $Q_0\in Q(A,\OO_{I_1})$ such that $Q_0(I_1)\not= Q_0(I_2)$ then
$$\max_{I\in C}|Q(A,\OO_{I})|\ge \E_{I\in C}[|Q(A,\OO_{I})|]\ge \log |C|.$$
That is, the worst-case query complexity and the average-case query complexity of $A$ is at least $\log |C|$.
\end{lemma}
\begin{proof}For $I\in C$, consider the sequence of the queries that $A$ with the oracle ${\cal O}_{I}$ asks and let $s(I)\in \cup_{i=0}^\infty\{0,1\}^i$ be the sequence of answers.
The worst case query complexity and average-case query complexity of $A$ are $s(C):=\max_{I\in C}|s(I)|$ and $\bar s(C):=\E_{I\in C}[|s(I)|]$, respectively, where $|s(I)|$ is the length of $s(I)$.
We now show that for every two distinct $I_1$ and $I_2$ in $C$, $s(I_1)\not=s(I_2)$ and $s(I_1)$ is not a prefix of $s(I_2)$.
This implies that $\{s(I)\ |\ I\in C\}$ contains $|C|$ distinct strings such that no string is a prefix of another. Then by Lemma~\ref{string1}, the result follows.
Consider two distinct sets $I_1,I_2\subseteq [n]$. There is a query $Q_0\in Q(A,\OO_{I_1})$ such that $Q_0(I_1)\not=Q_0(I_2)$.
Consider the sequence of queries $Q_{1,1}, Q_{1,2}, \cdots$ that $A$ asks with the oracle $\OO_{I_1}$
and the sequence of queries $Q_{2,1}, Q_{2,2}, \cdots$ that $A$ asks with the oracle $\OO_{I_2}$. Since $A$ is deterministic,
$A$ asks the same queries as long as it gets the same answers to the queries. That is, if $Q_{1,i}(I_1)=Q_{2,i}(I_2)$ for all $i\le \ell$ then
$Q_{1,\ell+1}=Q_{2,\ell+1}$. Then, either we get in both sequences to the query $Q_0$ and then $Q_0(I_1)\not=Q_0(I_2)$ or some other query $Q'$ that is asked before $Q_0$ satisfies $Q'(I_1)\not=Q'(I_2)$. In both cases $s(I_1)\not=s(I_2)$ and $s(I_1)$ is not a prefix of $s(I_2)$.\qed
\end{proof}

\section{Lower Bounds}
In this section, we prove some lower bounds for the number of queries that are needed in order to estimate the number of defective items.

\subsection{Lower Bounds for Deterministic and Las Vegas algorithms}
For deterministic algorithms, we prove
\begin{theorem}\label{THD} Let $A$ be a deterministic adaptive algorithm that estimates the number of defective items $|I|=d$ up to a multiplicative factor of $1\pm \epsilon$. The query complexity of $A$ is at least
$$d\log\frac{(1-\epsilon)n}{d}-O(d).$$

In particular, for $\epsilon\le 1-1/n^{\lambda}$ where $0<\lambda<1$ is any constant, the problem of estimating the number of defective items with a deterministic adaptive algorithm is asymptotically equivalent to finding them.
\end{theorem}
\begin{proof}Consider the sequence of queries that $A$ with an oracle $\OO_{I}$ asks and let $s(I)\in \cup_{i=1}^\infty \{0,1\}^i$ be the string of answers. Consider the algorithm $A$ with the oracles $\OO_{I_1}$ and $\OO_{I_2}$ where $I_1$ and $I_2$ are any sets of sizes $|I_1|=d$ and
$|I_2|\ge d':=(d+1)(1+\epsilon)/(1-\epsilon)$. For $I_1$, $A$ outputs an integer $D_1$ where $(1-\epsilon)d\le D_1\le (1+\epsilon)d$ and for $I_2$, $A$ outputs an integer $D_2$ where $d(1+\epsilon)+(1+\epsilon)\le  D_2$. Therefore, $D_1\not=D_2$ and hence $s(I_1)\not=s(I_2)$. This shows that if $|I_1|=d$ and $s(I_1)=s(I_2)$ then $|I_2|\le d'-1$.

Now let $I'\subseteq X$ be any set of size $d$. Let ${\cal I}$ be the set of all sets $I\subset X$ of size $d$ that have the same sequence of answers as $I$, i.e., $s(I)=s(I')$. Let $J=\cup_{I\in {\cal I}}I$. We now prove that $s(J)=s(I')$. Suppose for the contrary that this is not true. Then since $I'\subseteq J$ there is a query $Q$ asked by $A$ where $Q(J)=1$ and $Q(I')=0$. Therefore there is $j\in J\backslash I'$ such that $Q(j)=1$ and $Q(I')=0$. Since $j\in J$ there must be $I''\in {\cal I}$ such that $j\in I''$ and then $Q(I'')=1$. This is a contradiction to the fact that $s(I')=s(I'')$. Therefore, $s(J)=s(I')$, and by the above argument, we must have $|J|\le d'-1$. Since ${\cal I}$ contains subsets of $J$ of size~$d$, we have
$$|{\cal I}|\le L:={d'-1\choose d}.$$ This shows that each string in $\{s(I) : |I|=d\}$ corresponds to at most $L$ sets of size~$d$. Therefore $\{s(I) : |I|=d\}$ contains at least
$$M:=\frac{{n\choose d}}{{d'-1\choose d}}$$ distinct strings, and since the algorithm is deterministic, no string is a prefix of another. By Lemma~\ref{string1}, the longest string is of length at least
$$C:=\log M=\log\frac{{n\choose d}}{{d'-1\choose d}}\ge d\log\frac{n}{d}-d\log \left(\frac{1}{1-\epsilon}\right)-O(d).$$
Since the length of the longest string is the worst case query complexity of the deterministic algorithm the result follows.\qed
\end{proof}

For randomized Las Vegas algorithms, we prove
\begin{theorem}\label{THLV} Let $A$ be a randomized Las Vegas adaptive algorithm that estimates the number of defective items $|I|=d$ up to a multiplicative factor of $1\pm \epsilon$. The expected query complexity of $A$ is at least
$$d\log\frac{(1-\epsilon)n}{d}-O(d).$$

In particular, for $\epsilon\le 1-1/n^{\lambda}$ where $0<\lambda<1$ is any constant, the problem of estimating the number of defective items with a randomized Las Vegas adaptive algorithm is asymptotically equivalent to finding them.
\end{theorem}
\begin{proof}
Let $X(I,r)=|Q(A(r),\OO_I)|$ be a random variable of the number of queries that $A$ asks with oracle $\OO_I$ and let $g(d)=\max_{|I|=d}\E_r[X(I,r)]$ be the expected number of queries. Notice that for a fixed $r$, $A(r)$ is a deterministic algorithm.
Consider $S_r=\{s_r(I): |I|=d\}$ where $s_r(I)$ is the string of answers of the deterministic algorithm $A(r)$ with an oracle $\OO_I$. Suppose $S_r=\{w_1,\ldots,w_t\}$ and $|w_1|\le |w_2|\le \cdots \le |w_t|$. Consider a partition $W_1\cup W_2\cup \cdots\cup W_t$ of the set of all sets of size $d$,  where $W_i=\{I:|I|=d, s_r(I)=w_i\}$. As in the proof of Theorem~\ref{THD}, there are at least $t\ge M$ distinct strings in $S_r$. Also, no string is a prefix of another string because the algorithm is deterministic. Also, as in the proof of Theorem~\ref{THD}, for all $i$,
$$|W_i|\le {d'-1\choose d}.$$
Then, since $|w_1|\le |w_2|\le \cdots \le |w_t|$ and by Lemma~\ref{string2},
\begin{eqnarray*}
\E_I[X(I,r)|r]&=& \frac{\sum_{i=1}^t |W_i|\cdot |w_i|}{{n\choose d}}\\
&\ge & \frac{\sum_{i=1}^{M} {d'-1\choose d}\cdot |w_i|}{{n\choose d}}\\
&= & \frac{\sum_{i=1}^{M} |w_i|}{M}
\ge \log M.\\
\end{eqnarray*}

Thus
$$\E_I[\E_r[X(I,r)]]=\E_r[\E_I[X(I,r)|r]]\ge \log M.$$
Therefore, there is $I_0$ such that
$g(d)\ge \E_r[X({I_0},r)]\ge \log M$.\qed
\end{proof}
\subsection{Lower Bounds for Monte Carlo Algorithms}
We now give three lower bounds for randomized Monte Carlo adaptive algorithms.

Before presenting the first lower bound, it is important to note that when $\epsilon = \Theta(1/n)$, the algorithm that queries each item individually requires $n = O(1/\epsilon)$ queries. Therefore, we can assume that $\epsilon > 2/n$.

\begin{theorem}\label{TH1}
Let $2/n<\epsilon<1/2$, $d\ge 1/\epsilon$ and $\epsilon^\lambda\ge \delta\ge 1/n^{\lambda'}$ where $\lambda,\lambda'>1$ are any constants.
Let $A$ be a randomized Monte Carlo adaptive algorithm that estimates
the number of defective items up to a multiplicative factor of $1\pm \epsilon$. Algorithm $A$ must ask at least
$$\Omega\left(\frac{1}{\epsilon}\log\frac{1}{\delta}\right)$$ queries.
\end{theorem}
\begin{proof} It is enough to prove the result for $\epsilon^\lambda\ge \delta\ge 1/(n+2)$. This is because, under the assumption of such a result, any algorithm $A$ that has a failure probability of at most $\delta'$ where $1/(n+2)\ge \delta'\ge1/n^{\lambda'} $ also has a failure probability of at most $\delta:=1/(n+2)$, and therefore, the query complexity of $A$ is at least $\Omega((1/\epsilon)\log(1/\delta))=\Omega((1/\epsilon)\log(1/\delta'))$. 

Let $\epsilon^\lambda\ge \delta\ge 1/(n+2)$. Let $A(r)$ be a randomized Monte Carlo adaptive algorithm that, with probability at least $1-\delta$,
estimates
the number of defective items $|I|$ up to a multiplicative factor of $1\pm \epsilon$ where $r$ is the random seed of the algorithm.
Then for $|I|\in \{d',d'+1\}$ where $d'=\max(\lfloor 1/\epsilon\rfloor-2,1)<d$, it determines exactly $|I|$ with probability at least $1-\delta$. 
Let $X(I,r)$ be a random variable that is equal to $1$ if $A(\OO_I,r)\not= |I|$ and $0$ otherwise. Then for any $I\subseteq [n], |I|\in\{d',d'+1\}$ we have $\E_r[X(I,r)]\le \delta$. Let $m=\lfloor 1/(2\delta)\rfloor+d'-1\le n$.
Consider any $J\subseteq [m]$, $|J|=d'$. For any such $J$, let
$$Y_J(r)=X(J,r)+\sum_{i\in [m]\backslash J} X(J\cup\{i\},r).$$ Then for every $J\subseteq [m]$ of size $d'$,
$\E_r\left[Y_J(r)\right] \le (m-d'+1)\delta\le \frac{1}{2}.$ Therefore for a random uniform $J\subseteq [m]$ of size $d'$ we have
$\E_r[\E_J[Y_J(r)]]=\E_J[\E_r[Y_J(r)]]\le 1/2$. Thus, there is $r_0$ such that
for at least half of the sets $J\subseteq [m]$, of size $d'$, $Y_J(r_0)=0$. Let $C$ be the set of all $J\subseteq [m]$, of size $d'$, such that $Y_J(r_0)=0$. Then $$|C|\ge \frac{1}{2}{m\choose d'}=\frac{1}{2}{\lfloor 1/(2\delta)\rfloor+d'-1 \choose d'}.$$

Consider the deterministic algorithm $A(r_0)$.
We claim that for every two distinct $J_1,J_2\in C$, there is a query $Q\in Q(A(r_0),\OO_{J_1})$ such that $Q(J_1)\not=Q(J_2)$. If this is true then, by Lemma~\ref{mainlemma}, the query complexity of $A(r_0)$ is at least
$$\log |C|\ge \log \frac{1}{2}{\lfloor 1/(2\delta)\rfloor+d'-1 \choose d'}\ge d'\log \frac{1}{2d'\delta}-1=\Omega\left(\frac{1}{\epsilon}\log\frac{1}{\delta}\right).$$

We now prove the claim.
Consider two distinct $J_1,J_2\in C$. Since $|J_1|=|J_2|$ there exists $j\in J_2\backslash J_1$. Since $Y_{J_1}(r_0)=0$ we have $X(J_1,r_0)=0$ and $X(J_1\cup \{j\},r_0)=0$ and therefore $A(\OO_{J_1},r_0)=d$ and $A(\OO_{J_1\cup\{j\}},r_0)=d+1$. Thus, by Lemma~\ref{Trivial}, there is a query $Q_0\in Q(A(r_0),\OO_{J_1})\cap Q(A(r_0),\OO_{J_1\cup \{j\}})$ for which $Q_0(J_1)=0$ and $Q_0(J_1\cup\{j\})=1$. Therefore
$Q_0(\{j\})=1$ and then $Q_0(J_1)=0$ and $Q_0(J_2)=1$.\qed
\end{proof}

The following is the second lower bound for randomized Monte Carlo adaptive algorithms
\begin{theorem}\label{TH4p}Let $A$ be a randomized Monte Carlo adaptive algorithm that estimates the number of defective items with any\footnote{The constant $7/9$ can be substituted with any constant that is less than $1$.} $\epsilon<7/9$ and probability at least $1-\delta>1/2$. The query complexity of $A$
is at least
$$\log\log d-O(1).$$
\end{theorem}\label{Thss}
\begin{proof} Let $A$ be a randomized Monte Carlo algorithm that estimates $|I|\le d$ with probability at least $1-\delta$. Consider the class of sets of defective items $C=\{[8^i]|i=1,2,\ldots,\log d/3\}$. Since $(1+\epsilon)8^i<(1-\epsilon)8^{i+1}$, the algorithm can, with probability at least $1-\delta$, determine exactly the size of $I\in C$. 

For $I\in C$, let $X(I,r)$ be a random variable where $X(I,r)=1$ if $A(\OO_I,r)\not=|I|$ and $0$ otherwise. Then for every $j$, $\E_r[X([8^j],r)]\le \delta$. Now for a random uniform $[8^j]\in C$, we have $\E_r[\E_j[X([8^j],r)]]=\E_j[\E_r[X([8^j],r)]]\le \delta$. Therefore, there is a seed $r_0$ such that $\E_j[X([8^j],r_0)]\le \delta$. This implies that for at least $t:=(1-\delta)(\log d/3)$ sets $J:=\{[8^{j_1}],\ldots, [8^{j_t}]\}\subseteq C$ the deterministic algorithm $A(r_0)$  determines exactly $|I|$ provided that $|I|\in J$. Therefore, as in the above proofs, $A(r_0)$ asks at least
\begin{eqnarray}\label{jfjd}
\log t=\log\log d+\log(1-\delta)-\log 3\ge \log\log d-3.
\end{eqnarray} queries.\qed
\end{proof}
\subsection{Lower Bounds for Randomized Monte Carlo Algorithm with Expected Complexity}
We now consider randomized algorithms with success probability at least $1-\delta$ and $g(n,|I|,\delta)$ expected number of queries.

We first prove
\begin{theorem}\label{TH11}
Let $2/n<\epsilon<1/4$, $d\ge 1/\epsilon$ and $\epsilon^\lambda\ge \delta\ge 1/n^{\lambda'}$ where $\lambda,\lambda'>1$ are any constants.
Let $A$ be a randomized adaptive algorithm that estimates
the number of defective items up to a multiplicative factor of $1\pm \epsilon$. The expected number of queries of $A$ is at least
$$\Omega\left(\frac{1}{\epsilon}\log\frac{1}{\delta}\right).$$
\end{theorem}
\begin{proof} As in Theorem~\ref{TH1} we may assume that $\epsilon^\lambda\ge \delta\ge 1/(2n+4)$. Let $A(r)$ be a randomized algorithm that
estimates
the number of defective items up to a multiplicative factor of $1\pm \epsilon$ where $r$ is the random seed of the algorithm.
Then for and $|I|\in \{d',d'+1\}$ where $d'=\lfloor 1/\epsilon\rfloor-2$, it determines exactly $|I|$ with probability at least $1-\delta$. Let $X(I,r)$ be a random variable that is equal to $1$ if $A(\OO_I,r)\not= |I|$ and $0$ otherwise. Then for any $I\subseteq [n]$,\ $\E_r[X(I,r)]\le \delta$. Let $m=\lfloor \tau/\delta\rfloor+d'-1\le n$ where $\tau=1/4>\delta$ is a constant that will be determined later.
Consider any $J\subseteq [m]$, $|J|=d'$. For any such $J$, let
$$Y_J(r)=X(J,r)+\sum_{i\in [m]\backslash J} X(J\cup\{i\},r).$$ Then for every $J\subseteq [m]$ of size $d'$,
$\E_r\left[Y_J(r)\right] \le (m-d'+1)\delta\le \tau.$ Therefore for a random uniform $J\subseteq [m]$ of size $d'$ we have
$\E_r[\E_J[Y_J(r)]]=\E_J[\E_r[Y_J(r)]]\le \tau$. Let $\eta=1/2>\tau$ be a constant that will be determined later. By Markov's inequality, for random $r$, with probability at least $1-\tau/\eta$,
at least $1-\eta$ fraction of the sets $J\subseteq [m]$, of size $d'$, $Y_J(r)=0$. Let $R$ be the set of such $r$. Then $\Pr_r[R]\ge 1-\tau/\eta$. Let $r_0\in R$. Let $C_{r_0}$ be the set of all $J\subseteq [m]$, of size $d'$, such that $Y_J(r_0)=0$. Then $$|C_{r_0}|\ge (1-\eta){m\choose d'}=(1-\eta){\lfloor \tau/\delta\rfloor+d'-1 \choose d'}.$$

Consider the deterministic algorithm $A(r_0)$. As in Theorem~\ref{TH1},
for every two distinct $J_1,J_2\in C_{r_0}$, there is a query $Q\in Q(A(r_0),\OO_{J_1})$ such that $Q(J_1)\not=Q(J_2)$. Then by Lemma~\ref{mainlemma}, the average-case query complexity of $A(r_0)$ is at least
$$\log |C_{r_0}|\ge \log (1-\eta){\lfloor \tau/\delta\rfloor+d'-1 \choose d'}\ge d'\log \frac{\tau }{d'\delta}-\log\frac{1}{1-\eta}.$$

Let $Z(\OO_I,r)=|Q(A(r),\OO_I)|$. We have shown that for every $r\in R$, $$\E_{I\in C_{r}}[Z(\OO_I,r)]\ge d'\log \frac{\tau}{d'\delta}-\log\frac{1}{1-\eta}.$$
Therefore for every $r\in R$,
\begin{eqnarray*}
\E_{I}[Z(\OO_I,r)]&\ge& \Pr [I\in C_r]\cdot \E_{I}[Z(\OO_I,r)|I\in C_r]\\
&\ge& (1-\eta)\left( d'\log \frac{\tau}{d'\delta}-\log\frac{1}{1-\eta}\right).
\end{eqnarray*}
Therefore
\begin{eqnarray*}
\E_{I}\E_r[Z(\OO_I,r)]&=& \E_r\E_I[Z(\OO_I,r)]\\
&\ge &\Pr[r\in R]\cdot \E_r[\E_I[Z(O_I,r)]|r\in R]\\
&\ge& \left(1-\frac{\tau}{\eta}\right)(1-\eta) \left( d'\log \frac{\tau}{d'\delta}-\log\frac{1}{\eta}\right).
\end{eqnarray*}
Therefore, there is $I$ such that
\begin{eqnarray*}
\E_r[Z(\OO_I,r)]\ge \left(1-\frac{\tau}{\eta}\right)(1-\eta) \left( d'\log \frac{\tau}{d'\delta}-\log\frac{1}{\eta}\right).
\end{eqnarray*}
Now for $\eta=1/2$, $\tau=1/4$, $d'=\lfloor 1/\epsilon\rfloor-2$ and $\epsilon^\lambda\ge \delta$, we get
$$\E_r[Z(\OO_I,r)]= \Omega\left(\frac{1}{\epsilon}\log\frac{1}{\delta}\right).\qed$$
\end{proof}

In \cite{FalahatgarJOPS16}, Falhatgar et al. gave the following lower bound for $g(d,\delta)$.
We give another simple proof in the Appendix for slightly weaker lower bound.
\begin{theorem}\label{TH4} Let $A$ be a randomized adaptive algorithm that estimates the number of defective items $|I|=d$ up to a multiplicative factor of $1/2$ with probability at least $1-\delta$. The expected number of queries of $A$
is at least
$$(1-\delta)\log\log d.$$
\end{theorem}

\section{Upper Bound}
In this section, we prove some upper bounds.

\subsection{Upper Bounds for Deterministic and Las Vegas algorithms}
This section gives a tight upper bound for the deterministic algorithm that matches the lower bound in Theorem~\ref{THD}. The time complexity of this algorithm is linear in the size of the queries.

The following result will be used in this section.
\begin{lemma}\label{Find} \cite{ChengDX14,ChengDZ15,SchlaghoffT05} There is a deterministic adaptive algorithm, {\bf Find -Defectives}, which, without prior knowledge of $d$, asks $d\log(n/d)+O(d)$ queries and finds the defective items.
\end{lemma}

We now prove.

\begin{theorem}\label{TH3up} There is a deterministic adaptive algorithm that estimates the number of defective items $|I|=d$ up to  a multiplicative factor of $1\pm \epsilon$ and asks
$$d\log\frac{(1-\epsilon)n}{d}+O(d)$$ queries.
\end{theorem}
\begin{proof}The algorithm divides the set of items $X=[n]$ into $N=(1-\epsilon)n$ disjoint sets $X_1,\ldots,X_N$ where each set $X_i$ contains $1/(1-\epsilon)$ items. It then runs the algorithm {\bf Find-Defectives} in Lemma~\ref{Find} with $N$ items. For each query $Q\subseteq [N]$ in {\bf Find-Defectives}, the algorithm asks the query $Q'=\cup_{i\in Q}X_i$. By Lemma~\ref{Find}, the number of queries is
$$d\log\frac{N}{d}+O(d)=d\log\frac{(1-\epsilon)n}{d}+O(d).$$
Now since the $d$ defective items can appear in at most $d$ sets $X_i$ and at least $(1-\epsilon)d$ sets, the output of the algorithm is $D$ that satisfies $(1-\epsilon)d\le D\le d$.\qed
\end{proof}

\subsection{Upper Bounds for Randomized Monte Carlo Algorithm with Expected Complexity}
We now give a randomized algorithm that, for any constant $\epsilon$, its expected number of queries almost matches the lower bound in Theorem~\ref{TH4} and \ref{TH1}.
\begin{theorem}\label{TH7} For any constant $c>1$, there is a randomized algorithm that asks\footnote{The $\tilde O(\log (1/\delta))$ is $O((\log(1/\delta))(\log\log(1/\delta))$}
$$q=(1-\delta+\delta^c)\log\log d +O(\sqrt{\log\log d})+O\left(\frac{1}{\epsilon^2} \log\frac{1}{\delta}\right)+\tilde O\left(\log\frac{1}{\delta}\right)$$
expected number of queries and with probability at least $1-\delta$ estimates the number of defective items $d$ up to a multiplicative factor of $1\pm \epsilon$.
\end{theorem}
\begin{proof}
We first give an algorithm $A$ that asks
$$q'(\delta):=\log\log d +O(\sqrt{\log\log d})+O\left(\frac{1}{\epsilon^2}\log\frac{1}{\delta}\right)+\tilde O\left(\log\frac{1}{\delta}\right)$$
expected number of queries. We then define the following algorithm $B$: With probability $\delta-\delta^c$ output $0$ and with probability $1-(\delta-\delta^c)$ run algorithm $A$ with success probability of $1-\delta^c$.

The expected number of queries that $B$ asks is
$(1-\delta+\delta^c) q'(\delta^c)=q$ and the success probability is $1-\delta$.

We now give algorithm $A$. Algorithm $A$ is the same as the algorithm of Falahatgar et al.~\cite{FalahatgarJOPS16} but with different parameters. Their algorithm runs in 4 stages. In the first stage they give a procedure $\A_{{\rm FACTOR-}d}$ that finds an integer $D_1$ that with probability at least $1-\delta$ satisfies $d\le D_1\le 2d^2\frac{1}{\delta^2}\log\frac{1}{\delta}$. Procedure $\A_{{\rm FACTOR-}d}$ for $i=1,2,\cdots$, generates random queries $Q_i$ where each $j\in [n]$ is in $Q_i$ with probability $1-2^{-1/\Delta_i}$ and is not in $Q_i$ with probability $2^{-1/\Delta_i}$ where $\Delta_i=2^{2^i}$. It then asks the queries $Q_i$ for $i=1,2,\cdots$ and halts on the first query $Q_{i_0}$ that gets answer $0$. Then, it outputs $D_1=2\Delta_{i_0}\log\frac{1}{\delta}$.

Our procedure IMPROVED$\A_{{\rm FACTOR-}d}$ finds an integer $D_1'$ that with probability at least $1-\delta$ satisfies
$$d\le D_1'\le 2\left(\frac{2d}{\delta}\right)^{2^{2\sqrt{\log\log \frac{2d}{\delta}}+1}} \log\frac{1}{\delta}.$$
Procedure IMPROVED$\A_{{\rm FACTOR-}d}$ for $i=1,2,\cdots$, generates random queries $Q'_i$ where each $j\in [n]$ is in $Q'_i$ with probability $1-2^{1/\Delta'_i}$ where $\Delta'_i=2^{2^{i^2}}$, asks the queries $Q_i'$ and halts on the first query $Q'_{i_0}$ that gets answer $0$. Then, it outputs $D'_1=2\Delta_{i_0}\log\frac{1}{\delta}$. The expected number of queries in IMPROVED$\A_{{\rm FACTOR-}d}$ is
\begin{eqnarray}\label{ee1}\sqrt{\log\log D_1'}=O\left(\sqrt{\log\log \frac{d}{\delta}}\right).
\end{eqnarray}
The proof of correctness and the query complexity analysis is the same as in~\cite{FalahatgarJOPS16} and is sketched in the next subsection for completeness.

The second stage of Falahatgar et al. algorithm is the procedure $\A_{{\rm FACTOR}-1/\delta^2}$. The procedure $\A_{{\rm FACTOR}-1/\delta^2}$ is a binary search for $\log d$ in the logarithmic scale
of the interval $[1,D_1]$ - that is, in $[0,\log D_1]$. The procedure with probability at least $1-\delta$ returns $D_2$ such that $\delta^2 d\le D_2\le d/\delta^2$. This procedure is Monte Carlo. The number of queries is $\log\log D_1=\log\log \frac{d}{\delta}+O(1)$. The same procedure with the same analysis and proof of correctness works as well in our algorithm for the interval $[0,\log D_1']$. The procedure $\A_{{\rm FACTOR}-1/\delta^2}$, with probability at least $1-\delta$, returns $D'_2$ such that $\delta^2 d\le D'_2\le d/\delta^2$. The number of queries is
\begin{eqnarray}\label{ee2}\log\log D_1'=\log\log\frac{d}{\delta}+O\left(\sqrt{\log\log \frac{d}{\delta}}\right).
\end{eqnarray}
The third and fourth stages in~\cite{FalahatgarJOPS16} (and in our algorithm) are two procedures that, with an input $D_2'$, with probability at least $1-\delta$, estimate the number of defective items $d$ up to a multiplicative factor of $1\pm \epsilon$ with $O((1/\epsilon^2)\log(1/\delta))+\tilde O(\log(1/\delta))$ number of queries.

The expected number of queries is the sum of expressions in (\ref{ee1}), (\ref{ee2}) and  $O((1/\epsilon^2)\log(1/\delta))+\tilde O(\log(1/\delta))$ which is equal to $q'(\delta)$.\qed
\end{proof}

We note here that the best constant in the $O(\sqrt{\log\log d})$ is $2\sqrt{2}=2.828$ and can be obtained by the sequence $\Delta_i=2^{2^{i^2/2}}$.

\subsubsection{Analysis of the Algorithm.}
The following result is immediate.
\begin{lemma}\label{base}
Let $Q_\Delta$ be a random query where each $j\in [n]$ is in $Q_\Delta$ with probability
$1-2^{-1/\Delta}$ and is not in $Q_\Delta$ with probability $2^{-1/\Delta}$. Let
$I\subseteq [n]$ be a set of defective items of size $d$. Then for any $\Delta$ we have
$$\Pr[Q_\Delta(I)=0]=2^{-\frac{d}{\Delta}}$$ and for $\Delta> d$,
$$\Pr[Q_\Delta(I)=1]=1-2^{-\frac{d}{\Delta}}\le\frac{d}{\Delta}.$$
\end{lemma}
Let $\{\Delta_i\}_{i=1}^\infty$ be any sequence of numbers such that, $\Delta_1\ge 1$ and $\Delta_{i+1}/\Delta_i\ge 2$. Consider the algorithm that asks the query $Q_{\Delta_i}$ for $i=1,2,3,\ldots$ and stops on the first query $Q_{\Delta_{i_0}}$ that gets answer $0$. Let $$D=2\Delta_{i_0}\log\frac{2}{\delta}.$$

Since $\Delta_{i-1}\le \Delta_i/2$ and by Lemma~\ref{base},
\begin{eqnarray*}
\Pr[D<d]&=&\Pr\left[\Delta_{i_0}
<\frac{d}{2\log(2/\delta)}\right]\\
&\le&\sum_{i:\Delta_i<d/(2\log(2/\delta))}\Pr[Q_{\Delta_i}(I)=0]\\
&=&\sum_{i:\Delta_i<d/(2\log(2/\delta))}2^{-d/\Delta_i}\le \delta/2.
\end{eqnarray*}

Let $i_1$ be such that $\Delta_{i_1-1}\le 2d/\delta<\Delta_{i_1}$. Then, by Lemma~\ref{base},
\begin{eqnarray*}
\Pr\left[D>2\Delta_{i_1}\log\frac{2}{\delta}\right]&=&
\Pr[\Delta_{i_0}>\Delta_{i_1}]\\
&\le& \Pr[Q_{\Delta_{i_1}}(I)=1]\\
&\le & \frac{d}{\Delta_{i_1}}\le \delta/2.
\end{eqnarray*}
Since, $\Delta_{i+1}/\Delta_i\ge 2$, we have
$$\Pr[\Delta_{i_0}>\Delta_{i_1+k}]\le \frac{d}{\Delta_{i_1+k}}\le\frac{\delta}{2^{k+1}},$$
and therefore the expected number of queries is at most $i_1+2$.

This proves
\begin{lemma}\label{FAlg} Let $\{\Delta_i\}_{i=1}^\infty$ be any sequence of numbers such that, $\Delta_1\ge 1$ and  $\Delta_{i+1}/\Delta_i\ge 2$. Let $i_1$ be such that $\Delta_{i_1-1}\le 2d/\delta<\Delta_{i_1}$. The above algorithm asks at most $i_1+2$ expected number of queries and with probability at least $1-\delta$ outputs $D$ that satisfies $D\ge d$ and $D\le 2\Delta_{i_1}\log (2/\delta)$.

\ignore{Suppose we know some upper bound $D^*$ on $d$. Let $i_2$ be such that $\Delta_{i_2}>D^*$. The algorithm is also a Monte Carlo algorithm that asks at most $i_2$ queries.}
\end{lemma}

Now if we take $\Delta_i=2^{2^{i^2}}$ then $i_1\le \sqrt{\log\log (2d/\delta)}+1$ and
$$\Delta_{i_1}\le \left(\frac{2d}{\delta}\right)^{2^{2\sqrt{\log\log \frac{2d}{\delta}}+1}}.$$ Therefore
$$d\le D\le 2\left(\frac{2d}{\delta}\right)^{2^{2\sqrt{\log\log \frac{2d}{\delta}}+1}} \log\frac{2}{\delta}.$$
This gives the result in Theorem~\ref{TH7}.
\subsection{A Randomized Monte Carlo Algorithm}
In this section, we use a randomized Monte Carlo algorithm.

We now prove

\begin{theorem}\label{THl} There is a randomized Monte Carlo algorithm that asks
$$\log\log n +O\left(\frac{1}{\epsilon^2}\log\frac{1}{\delta}\right)+\tilde O\left(\log\frac{1}{\delta}\right)$$
queries and with probability at least $1-\delta$ estimates the number of defective items $d$ up to a multiplicative factor of $1\pm \epsilon$.
\end{theorem}
\begin{proof}
    We start from the second procedure of the Falahatgar et al. algorithm (the binary search) that asks $\log\log n$ queries. We get, with probability at least $1-\delta/2$, an integer $D$ such that $\delta^2d\le D\le d/\delta^2$. Then use the two next procedures of their algorithm that ask $O((1/\epsilon^2)\log(1/\delta))+\tilde O(\log(1/\delta))$ queries with a success probability of $1-\delta/2$.\qed  
\end{proof}

\ignore{Note: The above stages can even start from a much slower function. For example $\log^{**}n$ that is defined as $\log^{**}\alpha=1$ for $\alpha\le 2$ and $\log^{**}n=1+\log^{**}(\log^*n)$.}

\ignore{
\section{Open Problems}
The results in the table in Subsection~\ref{sectionTable} suggest the following open problems
\begin{enumerate}
\item Prove a lower bound $\Omega((1/\epsilon^2)\log (1/\delta))$ or find an randomized algorithm that asks $(1-\delta)\log\log d+O((1/\epsilon)\log (1/\delta))$ expected number of queries.
\item Prove the lower bound $\Omega(d)$ for number of queries in any randomized Monte Carlo algorithm when $n\to\infty$. A randomized Monte Carlo algorithm that asks $O(d\log d+d\log(1/\delta))$ queries follows from~\cite{Cheng11}.
\end{enumerate}}

%\bibliography{school01}
%\bibliographystyle{plain}
%$$\mbox{{\huge HERE}}$$
\newpage
\section{Appendix}
In this Appendix we give a simple proof of Theorem~\ref{TH4}.

\noindent
{\bf Theorem \ref{TH4} .} {\em Let $A$ be a randomized adaptive algorithm that estimates $d$ up to  multiplicative factor of $1/4$ with probability at least $1-\delta$. The expected number of queries of $A$
is at least
$$(1-\delta)(\log\log d-\log\log\log d-2) $$
}
\begin{proof} Let $A(r)$ be an adaptive algorithm that estimates $d$ up to  multiplicative factor of $1/4$ with probability at least $1-\delta$. Let $q(d)$ be the expected number of queries of $A(r)$.
Define a sequence of sets $I_1=[1],I_2=[2],\ldots, I_t=[2^t]$ where $2^t\le d$ and $2^{t+1}>d$. Then $t=\lfloor \log d\rfloor$. We restrict the inputs of $A$ to be only $I_j$ for some $j=1,\ldots,t$ and force $A$ to halt if it asks more than $q(d)/(1-\delta-\eta)$ queries where $\eta>0$ will be determined later. 
This new algorithm, denoted by $B$, is a Monte Carlo algorithm that finds exactly the size of $|I_j|$ with probability at least $1-(\delta+(1-\delta-\eta))=\eta$ and asks at most $q(d)/(1-\delta-\eta)$  queries. 
Therefore by Theorem~\ref{TH4p} (see (\ref{jfjd})), $q(d)/(1-\delta-\eta)\ge \log \log d+\log \eta$ and therefore for $\eta=(\ln 2)(1-\delta)/\log\log d$ we get
\begin{eqnarray*}
q(d)&\ge& (1-\delta-\eta)(\log\log d+\log \eta)\\
&\ge& (1-\delta)(\log\log d-\log\log\log d-2).\qed
\end{eqnarray*}
\end{proof}

\end{document}